\documentclass[11pt]{article}
\usepackage[T1]{fontenc} 
\usepackage{lmodern} 
\usepackage[utf8]{inputenc} 
\usepackage{cmap} 
\usepackage[pdfstartview=FitH,pdfpagemode=None,colorlinks,linkcolor=blue,filecolor=blue,citecolor=red,urlcolor=red]{hyperref}
\usepackage{fullpage}
\usepackage{multirow}
\usepackage{float}
\usepackage{comment}
\usepackage{amsmath,amsthm,amssymb,xcolor}
\usepackage[capitalize]{cleveref}
\usepackage{xspace}  
\usepackage{mathtools}
\usepackage{algorithm,bbm}
\usepackage{algorithmic}
\usepackage[margin=1in]{geometry}
\usepackage{fancyhdr}
\usepackage{graphicx}
\usepackage{enumerate}
\usepackage{framed}
\usepackage{tabu}
\usepackage{pgf, tikz}
\usetikzlibrary{quantikz}
\usepackage{hyperref}
\usepackage{braket}
\usepackage{mathtools}
\usetikzlibrary{arrows}
\newcommand{\eps}{\varepsilon}
\newcommand{\abs}[1]{\left| #1 \right|}
\newcommand{\vabs}[1]{\left\| #1 \right\|}

\newcommand{\pbra}[1]{\left( #1 \right)}
\newcommand{\sbra}[1]{\left[ #1 \right]}
\newcommand{\cbra}[1]{\left\{ #1 \right\}}
\renewcommand{\mid}{\,\middle\vert\,}

\newcommand{\poly}{\mathrm{poly}}
\newcommand{\polylog}{\mathrm{polylog}}

\newcommand{\Fbb}{\mathbb{F}}
\newcommand{\Nbb}{\mathbb{N}}
\newcommand{\Rbb}{\mathbb{R}}

\newcommand{\Acal}{\mathcal{A}}
\newcommand{\Bcal}{\mathcal{B}}

\newcommand{\Fcal}{\mathcal{F}}
\newcommand{\Gcal}{\mathcal{G}}

\newcommand{\Pcal}{\mathcal{P}}
\newcommand{\Scal}{\mathcal{S}}
\newcommand{\Hcal}{\mathcal{H}}

\newcommand{\ext}{\mathsf{Ext}}
\newcommand{\E}{\mathbb{E}}

\newcommand{\Tr}{\mathrm{Tr}}
 
\newtheorem{theorem}{Theorem}[section]
\newtheorem{lemma}[theorem]{Lemma}

\newtheorem{corollary}[theorem]{Corollary}
\newtheorem{claim}[theorem]{Claim}
\newtheorem{fact}[theorem]{Fact}
\newtheorem{definition}[theorem]{Definition}

\title{Eliminating Intermediate Measurements \\ using Pseudorandom Generators}
\date{}

\author{Uma Girish\thanks{Department of Computer Science, Princeton University. Research supported by the Simons Collaboration on Algorithms and Geometry, by a Simons Investigator Award and by the National Science Foundation grants No. CCF-1714779, CCF-2007462. \href{ugirish@cs.princeton.edu}{\texttt{ugirish@cs.princeton.edu}}} 
	\and Ran Raz\thanks{Department of Computer Science, Princeton University. Research supported by the Simons Collaboration on Algorithms and Geometry, by a Simons Investigator Award and by the National Science Foundation grants No. CCF-1714779, CCF-2007462. \href{ranr@cs.princeton.edu}{\texttt{ranr@cs.princeton.edu}}} }

\begin{document} 
\maketitle

\begin{abstract}
	We show that quantum algorithms of time $T$ and space $S\ge \log T$ with unitary operations and intermediate measurements can be simulated by quantum algorithms of time $T \cdot \poly (S)$ and space $ {O}(S\cdot \log T)$ with unitary operations and {\it without intermediate measurements}. The best results prior to this work required either $\Omega(T)$ space (by the deferred measurement principle) or $\poly(2^S)$ time~\cite{fr,grz}. Our result is thus a time-efficient and space-efficient simulation of algorithms with unitary operations and intermediate measurements by algorithms with unitary operations and \emph{without intermediate measurements}. 
	
	To prove our result, we study pseudorandom generators for quantum space-bounded  algorithms.
	We show that (an instance of) the INW pseudorandom generator for classical space-bounded algorithms~\cite{inw} also fools quantum space-bounded algorithms.
	More precisely, we show that for quantum space-bounded algorithms that have access to a read-once tape consisting of random bits, the final state of the algorithm when the random bits are drawn from the uniform distribution is nearly identical to the final state when the random bits are drawn using the INW pseudorandom generator. This result applies to general quantum algorithms which can apply unitary operations, perform intermediate measurements and reset qubits.
\end{abstract}

\section{Introduction}

\subsection{Eliminating Intermediate Measurements}

The main motivation for this work is the following fundamental question: What is the relative power of quantum algorithms {\bf with intermediate measurements} and quantum algorithms {\bf without intermediate measurements}?

The textbook's answer to this question is given by the {\it Principle of Deferred Measurements}. The principle states that delaying measurements until the end of a computation doesn't affect the output, as long as the qubits that were supposed to be measured do not further participate in the computation from that point on.
This gives a very simple method to eliminate all intermediate measurements in  quantum computations.
However, this simple method comes with a huge price in terms of the space needed to perform the computation, as qubits that were supposed to be measured cannot further participate in the computation.

More precisely, for any quantum algorithm of time $T$ and space $S\ge \log T$ with intermediate measurements, the principle of deferred measurements implies that it can be simulated by a quantum algorithm of time $T$ and space $S+T$ without intermediate measurements. Here, the overhead in space is potentially exponential. It was recently shown that there is also a simulation by algorithms of space $O(S)$ and time $\poly(T,2^S)$ without intermediate measurements~\cite{grz,fr}\footnote{Such a simulation was given in~\cite{grz} for algorithms that can use only quantum registers and independently in~\cite{fr} for the more general case of algorithms that may also use classical registers.}.  While the space overhead in these simulations is optimal, the time overhead is potentially exponential.


In this work, we show a simulation where the space dependence is ${O}(S\cdot\log T)$ and the time dependence is $T\cdot \poly(S)$. Our result applies to algorithms with unitary operations and intermediate measurements and simulates them by algorithms with unitary operations and \emph{without intermediate measurements}. Our result is thus a time-efficient and space-efficient simulation of algorithms with unitary operations and intermediate measurements by algorithms with unitary operations and without intermediate measurements. For example, our result implies that unitary quantum algorithms of polylogarithmic space and polynomial time are no less powerful than ones that can additionally perform intermediate measurements.

\begin{theorem}[Informal] \label{Thm:Informal1}
	Every quantum algorithm of time $T$ and space $S\ge \log T$ with unitary operations and intermediate measurements 
	can be simulated by a quantum algorithm of time $T\cdot S^2\cdot  \polylog (S)$ and space $ {O}(S \cdot \log T)$ with only unitary operations and no intermediate measurements.
\end{theorem}

 For quantum algorithms that can apply unitary operations, intermediate measurements \emph{and reset qubits}\footnote{The reset operation maps a qubit in an arbitrary state to the $\ket{0}$ state. If the algorithm has the ability to perform measurements, then the reset operation can be simulated with additional classical memory (on which classical operations are allowed). More precisely, to reset a qubit, we can measure it, swap the contents of this qubit with a classical bit in the $\ket{0}$ state by using two controlled-not operators and finally erase the contents of the classical register. Conversely, classical memory can be simulated using the reset operator. Thus, quantum algorithms with intermediate measurements and the ability to reset qubits correspond to quantum algorithms with both quantum and classical memory.}, we can trivially eliminate intermediate measurements using one ancilla qubit and with no overhead in time\footnote{We thank anonymous reviewers for this observation.}. This is because the measurement of a qubit can be simulated by copying the qubit to an ancilla qubit using a controlled-not gate and then resetting the ancilla qubit. 

\subsection{Pseudorandom Generators for Quantum Space-Bounded Computations}

It is well known that quantum measurements can be used to generate perfect random bits. The converse is also true and is fascinating in its own right:
An intermediate  measurement of a qubit can be implemented by unitary operations together with one random bit, as follows: with probability $1/2$ apply the identity matrix and with probability $1/2$ apply a reflection over $|1\rangle$ (in that qubit).
Thus, intermediate measurements are, in some sense, equivalent to randomness and in particular could be simulated by random bits.
From that perspective, it is very natural to try to derandomize and use pseudorandom bits, that is, to simulate intermediate measurements by pseudorandom bits.
In particular, we will use in this work pseudorandom generators for space-bounded computation, that are particularly suitable for our purpose.

Our main result is proved by studying (an instance of) the INW pseudorandom generator~\cite{inw} in the setting of quantum space-bounded algorithms. Let $S,T:\Nbb\to\Nbb$ be computable functions such that $S\ge \log T$. The INW pseudorandom generator $G$ for classical randomized algorithms of space $S$ and time $T$ is a function that takes inputs in $\{0,1\}^{N(M+1)}$ where $M=\Theta(\log T)$ and $N=\Theta( S  )$ (we refer to the input as the {\it seed} to $G$) and outputs a string in $\{0,1\}^T$, furthermore, this function is computable in space $O(S\cdot \log T)$ and time $T\cdot S^2\cdot \polylog (S)$. For any classical randomized algorithm of space $S$ and time $T$, the output of the algorithm when the random bits are drawn from the uniform distribution is nearly indistinguishable from the case when the random bits are drawn from the output of $G$ on a uniformly random seed~\cite{inw}. In this work, we show that a similar result holds for quantum algorithms of space $S$ and time $T$.

\begin{theorem}[Informal] ~\label{Thm:Informal2}
	Consider any quantum algorithm of space $S$ and time $T$ with arbitrary quantum operators that has access to a read-once tape consisting of random bits.  Then, the final state of the quantum algorithm when the random bits are drawn from the uniform distribution is nearly indistinguishable from the final state when the random bits are drawn from the output of $G$ on a uniformly random seed.
\end{theorem}


This result applies to quantum algorithms that can apply unitary operators, perform measurements or reset qubits. See \cref{maintheorem1} for a more formal statement. The INW pseudorandom generator is defined recursively using repetitive application of a randomness extractor. For the proof of Theorem~\ref{Thm:Informal2}, we use an instance of the INW generator, using a randomness extractor that was proved by Fehr and Schaffner to be resilient to quantum side information~\cite{fs}. Our proof that the generator fools quantum algorithms is a modification of the proof in the classical case, relying on tools and techniques from quantum information theory, in particular tools and techniques from~\cite{rennerthesis,fs}.

\subsection{Discussion and Additional Motivation}

Pseudorandom generators for classical space-bounded computation have been studied in numerous works (see for example~\cite{bns,nis,inw,nz,rr,rsv,brry,gr,mrt}). To the best of our knowledge,
pseudorandom generators for quantum space-bounded computation have not been studied before, possibly because, as mentioned above, quantum measurements can presumably generate perfect random bits, so the standard motivation of derandomizing randomized computations does not apply to the quantum case. Nevertheless, we believe that  the connection to eliminating intermediate measurements gives a strong motivation for studying pseudorandom generators for quantum space-bounded computation.

While eliminating intermediate measurements is our main motivation, we believe that pseudorandom generators for quantum space-bounded computation are interesting in their own right for various other reasons. First, any such generator implies in particular an indistinguishability result that gives a lower bound for the resources needed to achieve a certain computational task, which seems interesting from the point of view of complexity theory. Second, while in principle quantum computers may use measurements to generate perfect random bits, in reality this is hardly the case as quantum computers are likely to remain unreliable in the near future. In the last two decades, this motivated a large body of work on problems such as device-independent quantum cryptography, 
verifiable quantum computation
and certified randomness generation, all of which assume a setting where the quantum part of the device is unreliable and cannot be trusted. It's possible that pseudorandom generators for quantum space-bounded computations may find applications in these areas.
Finally, we find the question of how much true randomness is needed to simulate a quantum system fascinating.

Let us also mention that eliminating intermediate measurements is also interesting from the point of view of time-reversibility of computation.
Landauer introduced the concept of time-reversible computation and argued that any irreversible operation must be accompanied by entropy increase~\cite{Landauer} (see also~\cite{bennett}).
An interesting aspect of Theorem~\ref{Thm:Informal2} is that it shows that any quantum algorithm can be implemented using only time-reversible operations (except for the final measurement that gives the final output), with small overhead in time and space.

\subsection{Proof Overview}

We describe the proof of~\Cref{Thm:Informal2} in the classical case as in~\cite{inw}. This exploits the limited amount of information that is passed between successive states of the memory. Consider an algorithm $B$ of space $S$ and time $2\cdot T$ and assume that it uses a random bit at each time step. Let $B_0$ be the first half of the algorithm and $B_1$ be the second half. The algorithm $B_0$ uses a uniformly random string $U\sim\{0,1\}^T$ and $B_1$ uses a uniformly random string $U'\sim\{0,1\}^T$ that is independent of $U$. The only interaction between $B_0$ and $B_1$ is through the memory at time $T$. Since the memory is of at most $S$ bits, intuitively, the amount of information passed from $B_0$ to $B_1$ is at most $S$ bits. One may hope to replace the $T$ truly random bits used by $B_1$ with $T$ pseudorandom bits that essentially contain only $ S$ truly random bits. The idea is to apply an extractor to the random string used by $B_0$ and a uniformly random seed of length $\Theta(S)$. That is, for a suitable extractor $\ext:\{0,1\}^T\times \{0,1\}^{d}\to \{0,1\}^T$, we run the algorithm $B$ on the random string $(U,\ext(U,D))$ as opposed to $(U,U')$, where $D\sim \{0,1\}^{d}$ is uniformly random and $d=\Theta(S)$. This would effectively reduce the amount of randomness  from $2T$ bits to $T+\Theta(S)$ bits. The INW generator builds on this idea and recurses for $\Theta(\log T)$ steps, producing a pseudorandom generator of seedlength $\Theta(S\cdot \log T)$. We now justify the application of an extractor. Since there are at most $2^S$ possible memory states of the algorithm $B_0$, for most states $\mathcal{C}$ reached by $B_0$ at time $T$, the uniform distribution $U_{\mathcal{C}}$ on all strings which make $B_0$ reach the state $\mathcal{C}$ at time $T$ has min-entropy at least $T-\Theta(S)$. Suppose the extractor $\ext$ works against min-entropy at least $T-\Theta(S)$. (Such extractors are known and well-studied.) Then, the distribution of $\ext(U_\mathcal{C},D)$ would be close to the uniform distribution over $\{0,1\}^{T}$. In particular, the final state of the algorithm $B_1$ when starting at the state $\mathcal{C}$ would be nearly identical whether run according to $U'$ or according to $\ext(U_{\mathcal{C}},D)$. Since this holds for most states $\mathcal{C}$ reached by $B_0$ at time $T$, the distribution of final state of $B$ is nearly identical, whether we use the random string $(U,U')$ or the pseudorandom string $(U,\ext(U,D))$. 

To extend this idea to quantum algorithms, we have to deal with memory that is an arbitrary quantum state. In particular, we cannot ``condition'' on the memory at a particular time step. We instead use extractors that are \emph{resilient to quantum side information}. These are extractors with the following property: Suppose for each string $u\in\{0,1\}^{T }$, we have some quantum state $\rho_u$ on $S$ qubits (this represents the memory of the first half of the algorithm when run on the string $u$), then the distribution of $(\ext(U),\rho_U)$ is close to the tensor product of the fully mixed state over $\{0,1\}^{T }$ and the state $\E_{u\sim \{0,1\}^{T }}[\rho_u]$. Such extractors have been studied and exhibited with seedlength $\Theta(S)$~\cite{fs}. We use these extractors and modify the proof from the classical case
to derive our result, relying on tools and techniques from quantum information theory, in particular tools and techniques from~\cite{rennerthesis,fs}.

To prove~\Cref{Thm:Informal1}, we make use of the aforementioned equivalence between intermediate measurements and random bits. Given a quantum algorithm with $T$ intermediate measurements, we consider the equivalent quantum algorithm with $T$ random bits. Consider the algorithm that generates a uniformly random string on $\Theta(S\cdot \log T)$ bits, computes the output of the INW generator on this random string and simulates the above algorithm on the output of this generator as opposed to $T$ uniformly random bits. \Cref{Thm:Informal2} implies that this step introduces negligible error. Note that the $\Theta(S\cdot \log T)$ random bits used by the algorithm can be simulated \emph{unitarily} using additional $\Theta(S\cdot \log T)$ space, by the principle of deferred measurements. The rest of the technical work is devoted to the analysis of the space and time complexity of the INW generator with regards to unitary quantum algorithms. For this, we make use of the property of the extractor in~\cite{fs} that for every fixed seed, the extractor is a bijection. This is useful with regards to unitary computation. We show that each step of the recursion tree involved in computing the INW generator can be executed reversibly and efficiently by unitary quantum algorithms of small space. This completes the proof of our simulation result.

\subsection{Organization}

In \Cref{sec:preliminaries}, we formally define the various models of quantum computation with bounded space and time. In \Cref{sec:inw}, we define the INW pseudorandom generator and study its time and space complexity with respect to unitary quantum algorithms. We state our main theorem in \cref{sec:main}. We prove Theorem 1.2 (\cref{maintheorem1}) and  Theorem 1.1 (\cref{maintheorem2}) in \cref{sec:maintheorem1}  and \cref{sec:maintheorem2} respectively.

\section{Notation \& Preliminaries}
\label{sec:preliminaries}

For a mathematical statement $\mathcal{P}$, we use $\mathbbm{1}_{\mathcal{P}}\in\{0,1\}$ to refer to a boolean value which is 1 if $\mathcal{P}$ is true and 0 if $\mathcal{P}$ is false.

\subsection{Probability Distributions}
Let $\Sigma$ be an alphabet and $D$ be a probability distribution over $\Sigma$. We use $x\sim D$ to denote $x$ sampled according to $D$. For a subset $S\subseteq \Sigma$, we use $x\sim S$ to denote $x$ sampled according to the uniform distribution on $S$. For a multiset $S$ of $\Sigma$, we use $x\sim S$ to denote $x$ sampled with probability proportional to the number of times it occurs in $S$. Let $N\in \Nbb$. We use $U_N$ to denote the uniform distribution on $\{0,1\}^N$. For a function $G:\Sigma\to \Rbb^N$, we use $G(D):=\E_{x\sim D}[G(x)]$ to denote the expected output of $G$ when the inputs are drawn according to $D$.

\subsection{Quantum States}

Let $\mathcal{H}_m$ be a Hilbert space of dimension $2^m$. This is a vector space defined by the $\mathbb{C}$-span of the orthonormal basis $\{\ket{x}:x\in\{0,1\}^m\}$, that is, every element in this space is a unique complex combination of the vectors $\ket{x}$, where $x$ is a bit string in $\{0,1\}^m$. We use $\ket{0^m}$ to denote the state $\ket{(0,\ldots,0}$ on $m$ qubits. We omit the subscript $m$ when it is implicit. The complex conjugate of the vector $\ket{x}$ is denoted by $\bra{x}$. Let $\mathcal{P}(\mathcal{H}_m)$ be the set of all non-negative operators on $\mathcal{H}_m$, that is positive semidefinite matrices in $\mathbb{C}^{2^m\times 2^m}$. Let $\mathcal{S}(\mathcal{H}_m)$ be set of density operators on $\mathcal{H}_m$, that is, matrices in $\mathcal{P}(\mathcal{H}_m)$ with trace 1. We typically use $\rho$ to refer to elements of $\mathcal{S}(\mathcal{H}_m)$. Every element of $\mathcal{P}(\mathcal{H}_m)$ can be expressed uniquely as a complex combination of $\ket{i}\bra{j}$ where $i,j\in\{0,1\}^m$. We denote the identity matrix on $2^m\times 2^m$ by $\mathbb{I}_m$, and we omit the subscript if it is implicit. The state of a quantum system $M$ on $m$ qubits is described by a density operator $\rho_M\in \mathcal{S}(\mathcal{H}_m)$. The completely mixed state on $m$ qubits is described by $\frac{\mathbb{I}_m}{2^m}$. A classical state is a diagonal density operator in $\mathcal{S}(\mathcal{H}_m)$. 

Let $X$ be a system of $n$ bits and $S$ be a system of $s$ qubits. A classical-quantum state $\rho_{XS}$ is a state of the form $ \sum_{x\in \{0,1\}^n} \ket{x}\bra{x} \rho_x$ where $\rho_x\in \Pcal(\mathcal{H}_s)$ and $\sum_{x\in \{0,1\}^n} \Tr( \rho_x)=1$. We say that $\rho_{XS}$ is classical on $X$. We use $\rho_X=\ket{x}\bra{x}\Tr(\rho_x)$ to denote the induced classical state on the qubits in $X$. Similarly, we use $\rho_S=\sum_{x\in \{0,1\}^n} \rho_x$ to denote the induced state on the qubits in $S$. In this paper, it is often the case that the induced state on $X$ corresponds to the uniform distribution over $\{0,1\}^n$. In this case, we use $\E_{x\sim\{0,1\}^n} \sbra{\ket{x}\bra{x}\rho_x}$ to denote the state $\rho_{XS}$ where $\rho_x\in \Scal(\Hcal_s)$ for all $x\in \{0,1\}^n$. We say that $\rho_{XS}$ is uniform on $X$.

\subsection{Quantum State Evolution}

The evolution of a quantum state is described by a linear transformation $E:\mathcal{S}(\mathcal{H})\rightarrow \mathcal{S}(\mathcal{H})$ which is CPTP (that is, completely positive and trace preserving). In our work, we focus on transformations between vector spaces of the same dimension. We focus on the following quantum operations. 
\begin{itemize}
	\item Unitary Operators: An arbitrary unitary map $U:\Hcal\rightarrow \Hcal$ defines a CPTP map which maps $\rho\in\Scal(\Hcal)$ to $U\rho U^\dagger$. We make use of the following two unitary matrices, these are universal for unitary quantum computation~\cite{shi}.
	\begin{itemize}
	\item Hadamard:  $H:\Hcal_1\to\Hcal_1,H=\frac{1}{\sqrt{2}} \begin{bmatrix} 1 & 1 \\ 1 & -1 \end{bmatrix}$.
	\item Toffoli: $U:\Hcal_3\rightarrow \Hcal_3$  maps basis states $\ket{i,j,k}$ to $\ket{i,j,k\oplus i\cdot j}$ for $i,j,k\in \{0,1\}$. 
	\end{itemize}
	These operations naturally extend to operations on a larger Hilbert space by acting on a subset of qubits.
	\item The measurement operator $M$ on the first qubit maps the state $\rho=\sum_{i,j\in\{0,1\}}  \ket{i}\bra{j} \rho_{i,j}$ to the state $\sum_{i\in\{0,1\}}  \ket{i}\bra{i}\rho_{i,i}$ for all $\rho\in\Scal(\Hcal)$.
	\item The reset operator $R$ on the first qubit maps the state $\rho=\sum_{i,j\in\{0,1\}}  \ket{i}\bra{j} \rho_{i,j}$ to the state $ \ket{0}\bra{0} \pbra{\sum_{i\in\{0,1\}}\rho_{i,i} }$ for all $\rho\in\Scal(\Hcal)$. 
\end{itemize}

\subsection{Distance between States}

For any matrix $M$, we denote by $\|M\|_1$ its trace norm, that is $\|M\|_1:=\mathrm{Tr}(\sqrt{MM^\dagger}).$ Let $\rho,\sigma\in\mathcal{S}(\mathcal{H})$ be density operators. We define the trace distance between $\rho$ and $\sigma$ to be $d_1(\rho,\sigma):=\frac{\|\rho-\sigma \|_1}{2}.$
We will use the following standard facts about the trace distance. Firstly, the trace distance satisfies triangle inequality. Secondly, the trace distance between $\rho$ and $\sigma$ is equal to the maximum probability with which these states can be distinguished using a projective measurement $E,\mathbb{I}-E$ onto two subspaces. Thirdly, quantum operations cannot increase the trace distance. More formally,
\begin{fact}     For all $\rho_1,\rho_2,\rho_3\in \Scal(\Hcal)$,	\label{tracedistancetriangle} 
	$d_1(\rho_1,\rho_3)\le d_1(\rho_1,\rho_2)+d_1(\rho_2,\rho_3).$
\end{fact}
\begin{fact}  For all $\rho,\sigma\in \Scal(\Hcal)$,	\label{tracedistancedefinition} 
	$d_1(\rho,\sigma)= \underset{0\preceq E \preceq \mathbb{I}}{\max} \mathrm{Tr}(E(\rho-\sigma)).$
\end{fact}
\begin{fact}   Let $\rho,\sigma$ be two quantum states in $\mathcal{S}(\mathcal{H})$ and $E$ be a quantum operation on $\mathcal{S}(\mathcal{H})$.  Then,
	\[ d_1(E(\rho),E(\sigma)) \le d_1(\rho, \sigma). \]
	\label{tracedistanceproperty}
\end{fact}

\subsection{Quantum Space Bounded Computation}

There are two ways to define models of computation, one using uniform families of circuits and one using Turing machines~\cite{watrous}. Typically, these models are computationally equivalent, both in the classical case and the quantum case. For instance, it is known that polynomial time quantum Turing machines are equivalent to uniform families of quantum circuits of polynomial size~\cite{yao}. It is also known that logspace quantum Turing machines are equivalent to uniform families of quantum circuits of logarithmic width~\cite{fr}. With respect to computation with constraints on both time and space, few results are known. In the classical case, every deterministic Turing machine of space $S$ and time $T$ can be simulated by a logspace-uniform family of classical circuits of size $\poly(T)$ and width $O(S)$, conversely, every logspace-uniform family of classical circuits of size $T$ and width $S$ can be simulated by deterministic Turing machines of space $O(S\cdot\log T)$ and time $\poly(T,S)$~\cite{nick}. In particular, Turing machines of polylogarithmic space and polynomial time are computationally equivalent to logspace-uniform families of circuits of polynomial size and polylogarithmic width. Thus, without loss of generality, we can define classical algorithms of bounded space and bounded time as logspace-uniform families of circuits of bounded size and bounded width. We are unaware of such a result for quantum Turing machines. Nevertheless, in our paper, we define quantum algorithms based on the latter model, as it is easier to work with. We define space $S$ time $T$ quantum algorithms as logspace-uniform families of quantum circuits with $S$ qubits and $T$ operators. The formal definition is as follows.

\paragraph*{Quantum Algorithms:}

Let $\Gcal_U$ be a universal family of unitary operators for quantum computation, for instance, the Hadamard gate and the Toffoli gate. Let $\Gcal_M=\Gcal_U \cup \{ M\}$ (respectively $\Gcal_R=\Gcal_U\cup\{R\}$) include the measurement operator (respectively the reset operator) in addition to the previous operators.  

 Let $S,T:\Nbb\rightarrow \Nbb$ be computable functions and $\Gcal\in \{\Gcal_U,\Gcal_M,\Gcal_R,\Gcal_M\cup\Gcal_R\}$. A space $S=S(n)$ time $T=T(n)$ quantum algorithm $Q$ with input $x\in\{0,1\}^n$ consists of the initial state $\rho_0:=\ket{0^S}\bra{0^S}$ and a sequence of $T$ operators $E_{i,x}:\Scal(\Hcal_S)\rightarrow \Scal(\Hcal_S),E_{i,x}\in \Gcal$ for $i\in [T]$. Furthermore, this sequence is logspace uniform, that is, there is a classical deterministic logspace Turing Machine which on input $x$ outputs this sequence of operators along with the qubits on which they act. We use $Q^{\rho_0}(x):=E_{T,x} \cdots  E_{1,x}(\rho_0)=\left(\prod_{i=1}^T E_{i,x}\right)(\rho_0)\in \Scal(\Hcal_S)$ to refer to the final state of the algorithm. The output of the algorithm is defined to be the outcome on measuring the first qubit of the final state. Let $\Fcal=\{f_n:\{0,1\}^n\to\{0,1\} \}_{n\in\Nbb}$ be a family of partial boolean functions. We say that $\Fcal$ is computable by an algorithm if the algorithm on input $x\in \{0,1\}^n$ outputs $f_n(x)$ with probability at least $\frac{2}{3}$ (whenever $f_n(x)$ is well defined). For families of functions with  output of arbitrary length, we say that $\Fcal$ is computable by an algorithm if the algorithm on input $i\in\Nbb,x\in\{0,1\}^n$, computes the $i$-th bit of $f_n(x)$. 
 
An algorithm is said to be {\it   unitary} if the operators are from $\Gcal_U$, {\it   purely quantum} if the operators are from $\Gcal_M$ and simply {\it   quantum} if the operators are from $\Gcal_R\cup\Gcal_M$. The algorithm is said to have {\it   no intermediate measurements} if the operators are from $\Gcal_R$. We say that an algorithm $\Bcal$ simulates an algorithm $\Acal$ with error $\eps$ if for every $x\in\{0,1\}^*$ and $b\in\{0,1\}^*$ the probability that $\Acal(x)$ outputs $b$ and the probability that $\Bcal(x)$ outputs $b$ differ by at most $\eps$.
 
\paragraph*{Quantum Branching Programs with Randomness}
We now consider a model of quantum computation which is equipped with an additional classical randomness tape. The random string $r\in \{0,1\}^T$ in the randomness tape is read exactly once from left to right and on reading the bit $r_i$ at the $i$-th step, the program applies a quantum operator $E_{i,r_i,x}$. More formally, a quantum branching program $B^{\rho_0}_x(r)$ of space $S$ with input $x\in\{0,1\}^n$ and with {\it $T$ bits of randomness} consists of a sequence of $2T$ operators $E_{i,0,x},E_{i,1,x}:\mathcal{S}(\mathcal{H}_S)\rightarrow \mathcal{S}(\mathcal{H}_S)$ for $i\in[T]$, each of which is in $\Gcal_{R}$. Furthermore, this sequence is logspace uniform, that is, there is a classical deterministic logspace Turing Machine which on input $x$ outputs this sequence of operators. The branching program also has an initial state $\rho_0\in\Scal(\Hcal_S)$ (which is typically the all zeroes state) and takes an input string $r\in\{0,1\}^T$ in the randomness tape. We use $B^{\rho_0}(r):=E_{T,r_T,x}\cdots E_{1,r_1,x}(\rho_0)\in\Scal(\Hcal_S)$ to refer to the final state of the branching program on the  string $r$. For any distribution $D$ on $\{0,1\}^T$ we will denote by $B^{\rho_0}(D)$ the average final state $\mathbb{E}_{r\sim D} \sbra{ B^{\rho_0}(r)}$. The output of the branching program is defined to be the outcome on measuring the first qubit of $B^{\rho_0}(U_T)$. As before, we say that a branching program computes a family $\Fcal=\{f_n:\{0,1\}^n\to \{0,1\}\}_{n\in \Nbb}$ of functions if on input $x\in \{0,1\}^n$, the branching program outputs $f_n(x)$ with probability at least $\frac{2}{3}$ (whenever $f_n(x)$ is well defined). For families of functions with outputs of arbitrary length, we use the same definition as before.

\subsection{Quantum Entropic Quantities}

Let $XS$ be a possibly correlated bipartite quantum system, where $X$ is on $n$ qubits and $S$ is on $s$ qubits. Let $\rho_{XS}=\sum_x \ket{x}\bra{x}\rho_x$ be a classical-quantum state which is classical on $X$ and let $\rho_S\in \Scal(\Hcal_s)$.  We now define the {\it  min-entropy} of the state $\rho_{XS}$.
\[ H_{\min}(\rho_{XS}):= \sup \Big\{ \lambda \in \Rbb\text{ such that }   \rho_{XS} \preceq \frac{\mathbb{I}_{n+s}}{2^\lambda} \Big\} \]
The {\it  conditional min-entropy of} $\rho_{XS}$ {\it  relative to} $\sigma_S$ is defined as
\[ H_{\min}(\rho_{XS}|\sigma_S):= \sup \Big\{ \lambda\in \Rbb \text{ such that }   \rho_{XS} \preceq \frac{\mathbb{I}_n}{2^\lambda}\otimes \sigma_S \Big\}. \]
The {\it  conditional min-entropy} of $\rho_{XS}$ given $S$ is defined as
\[H_{\min}(\rho_{XS}|S) = \sup_{\sigma_S\in\Scal(\Hcal)} H_{\min} (\rho_{XS}|\sigma_S).\]
For any invertible $\sigma_S\in \Scal(\Hcal_s)$, the {\it conditional collision entropy} of $\rho_{XS}$ {\it relative to $\sigma_S$} is defined as
\[ H_2(\rho_{XS}|\sigma_S):= -\log \Tr\pbra{\pbra{(\mathbb{I}\otimes \sigma_S^{-1/4}) \rho_{XS}(\mathbb{I}\otimes \sigma_S^{-1/4}) }^2}. \]
The {\it conditional collision entropy} of $\rho_{XS}$ is defined as 
\[ H_2(\rho_{XS}|S):=\sup_{\sigma_S\in\Scal(\Hcal)} H_2(\rho_{XS}|\sigma_S). \]

It has been shown that the conditional collision entropy is bounded from below by the conditional min-entropy. The following fact appears as Remark 5.3.2 in~\cite{rennerthesis}.

\begin{fact} \label{mintwoentropy} For any classical-quantum state $\rho_{XS}$ and a quantum state $\sigma_S$, we have $ H_{\min}(\rho_{XS}|\sigma_S) \le H_2(\rho_{XS}|\sigma_S)$.
\end{fact}
It is known that conditioning on $s$ qubits cannot decrease the min-entropy by more than $s$. The following fact follows from Lemma 3.1.10 and Definition 3.1.2~\cite{rennerthesis}.  
\begin{fact} \label{conditioning} For any classical-quantum state $\rho_{XS}$, we have $ H_{\min} (\rho_{XS}|S) \ge H_{\min}(\rho_{XS})-s $. \end{fact}

\subsection{Extractors Resilient to Quantum Side Information}

We make use of $\delta$-biased spaces over $\{0,1\}^n$. There are many known constructions~\cite{naornaor,alonetal} of such spaces. We make use of a simpler construction based on the work of~\cite{alonetal}. We do this mainly so that we can easily argue about the time and space complexity.

We use finite fields of characteristic two to define our  $\delta$-biased spaces. We will use the fact that given any $m\in \Nbb$ such that $m=2\cdot 3^i$ for some $i\in\Nbb$, we can efficiently construct the finite field of characteristic two of size $2^m$. Futhermore, addition and multiplication of two elements in this field can be done in space and time at most $m\cdot \polylog (m)$. Similarly, raising elements of the field to $k$-th powers can be done in space $m\cdot \polylog (m)+ \log k$ and time $m\cdot \polylog( m)\cdot \log k$ by repeated squaring. The proofs of these facts can be derived from properties about finite fields and can be found in \cite{galoistextbook}. We defer this discussion to the appendix. We denote the finite field of size $2^m$ by $\mathbb{F}_{2^m}$. 

\paragraph*{$\delta$-biased spaces} Let $\mathcal{C}\subseteq\{0,1\}^n$ be a multiset and $\delta\ge0$. We say that $\mathcal{C}$ is a $\delta$-biased space if for all $S\subseteq [n]$, we have $\abs{\E_{x\sim \mathcal{C}} \sbra{\oplus_{i\in S}x_i} -\tfrac{1}{2}} \le \delta/2$. We make use of the following construction of $\delta$-biased spaces over $\{0,1\}^n$. This is implicit in the work of~\cite{alonetal}. While their construction obtains a better dependence on the field size, our variant is weaker but suffices for our purposes. 

Let $n\in \Nbb$. The $\delta$-biased space over $\{0,1\}^n$ is defined as follows. Let $m\in \Nbb$ be any integer such that $m\ge  \log(n/\delta)$ and $m=2\cdot 3^i$ for some $i\in [n]$.  Let $\langle \cdot,\cdot\rangle_2$ denote the inner product over $\mathbb{F}_2$. Define $A:\mathbb{F}_{2^m}\times \mathbb{F}_{2^m}\to \{0,1\}^n$ at $\alpha,\beta\in\mathbb{F}_{2^m}$ by  \[A(\alpha,\beta)=  \pbra{\langle 1,\beta\rangle_2,\langle \alpha,\beta\rangle_2, \ldots, \langle \alpha^{n-1},\beta\rangle_2 }.\] 
The proof of the following lemma follows from similar arguments as in~\cite{alonetal}.

\begin{lemma}\label{deltabias} Let $\Acal:=\cbra{\!\cbra{ A(\alpha,\beta) \mid \alpha,\beta\in\mathbb{F}_{2^{m }} }\!}$. Then, $\Acal$ is a $\delta$-biased space over $\{0,1\}^n$. Additionally, there is a classical deterministic algorithm which given input $\alpha,\beta\in\Fbb_2^m$ and $i\in[n]$, computes the $i$-th coordinate of $A(\alpha,\beta)$, furthermore, this algorithm uses space $m\cdot \polylog (m)+\log n$ and time $m\cdot \polylog( m)\cdot \log n$.
\end{lemma}

We now define weak quantum extractors, i.e., extractors that are resilient to quantum side information. 

\begin{definition}[Weak Quantum Extractor]~\cite{fs} Let $t,\eps\ge 0$. 
	A function $E:\{0,1\}^n \times \{0,1\}^d\rightarrow \{0,1\}^m$ is called a $(t,\eps)$-weak quantum extractor if the following holds. Let $\rho_{XS}=\E_{x\sim \{0,1\}^n}\sbra{\ket{x}\bra{x}\rho_x} \in \Hcal_{n+s}$ be any classical-quantum state that is classical and uniform on $X$. Suppose $H_2(\rho_{XS}|S)\ge t$, then
	\[ \vabs{ \E_{\substack{x\sim\{0,1\}^n\\y\sim \{0,1\}^d}}\sbra{\ket{E(x,y)}\bra{E(x,y)}\rho_x} - \frac{\mathbb{I}_m}{2^m}\otimes \E_{x\sim\{0,1\}^n}[ \rho_x]}_1 \le \eps.\] 
 The seedlength of the extractor is defined to be $d$.
\end{definition}

We make use of the following family of weak quantum extractors. 

\begin{theorem} \label{fsextractor} ~\cite{fs} Let $\delta>0, d,n\in \Nbb$, and $\Acal=\{a_1,\ldots,a_{2^d}\}$ be a $\delta$-biased space over $\{0,1\}^n$ of size $2^d$. Let $\ext:\{0,1\}^n\times \{0,1\}^d \to\{0,1\}^n$ be defined at $x\in\{0,1\}^n,i\in\{0,1\}^d$ by $ \ext(x,i):=a_i\oplus x.$ For $0\le t\le n$, $\ext$ is a $(t,\delta\cdot 2^{(n-t)/2})$-weak quantum extractor.  
\end{theorem}
\Cref{fsextractor} follows from Theorem 3.2 in~\cite{fs}. We remark that the result in~\cite{fs} is stated in terms of a $\delta$-biased {\it family} of distributions, however, we restrict ourselves to the case that there is one $\delta$-biased distribution. We derive the following corollary.

\begin{corollary}\label{fsextractorcorollary}
For any $n,t\in \Nbb$ and $\delta>0$, the function $\ext:\{0,1\}^n\times \{0,1\}^d\to \{0,1\}^n$ is a $(t,\delta)$-weak quantum extractor provided $d\ge \Theta(n-t + \log n+ \log(1/\delta) +O(1))$. Furthermore, there is a deterministic algorithm of space $O(d\cdot\polylog (d) + \log n)$ and time $d\cdot \polylog (d)\cdot \log n$ which computes any coordinate of the output of this extractor.  Additionally, for every fixed seed $y\in \{0,1\}^d$, the function $\ext_y:\{0,1\}^n\to \{0,1\}^n$ defined by $\ext_y(x)=\ext(x,y)$ is a bijection.
\end{corollary}

The property about the extractor being a bijection for every fixed seed turns out to be a useful property with regards to unitary (reversible) simulation.

\section{The INW Pseudorandom Generator}

\label{sec:inw}
We focus on a specific instantiation of the INW generator which uses the aforementioned extractors. We do this for two reasons: firstly, these extractors are known to be resilient to quantum side information; secondly, their time and space complexity is easier to analyze. Our construction is as follows. 

Fix parameters $T,S,\eps>0$ and let $N,M$ be integers such that
\[ N=\Theta( \log T +  S + \log (1/\eps) )\quad \quad \text{and}\quad \quad M=\lceil \log T\rceil. \]
Since we assume that $S\ge \log T$, we have $N=\Theta(S+\log(1/\eps))$. 
For each $i\in [M]$, let $\ext^{(i)}:\{0,1\}^{i N}\times \{0,1\}^N\to\{0,1\}^{iN}$ be as defined in \Cref{fsextractorcorollary} with parameters $n=iN,d=N,t=iN-S$ and $\eps=\eps/T^2$. We use $\ext^{(i)}_{s}(x)$ to denote the output of the extractor on input $x\in \{0,1\}^{iN}$ and seed $s\in \{0,1\}^N$. 

\paragraph*{Computational Complexity of Our Extractors} Let $i\in[M]$ and $j\in [iN]$. Note that computing the $j$-th coordinate of $\ext^{(i)}(z,s)$ where $z\in\{0,1\}^{iN}$ and $s\in\{0,1\}^N$ requires computing the $j$-th coordinate of $A(s)\in\{0,1\}^{iN}$. Due to~\Cref{fsextractorcorollary}, this can be done in space $N\cdot \polylog (N) + \log (iN)\le N\cdot \polylog (N)$ and time $N\cdot \polylog (N) \cdot \log (iN)\le N\cdot \polylog (N)$. This implies that for $i\le M$, the function $\ext^{(i)}$ is computable in space and time at most $N\cdot \polylog (N)$. We use the following facts about reversible simulation of deterministic computation. It is known that deterministic Turing Machines of space $S'$ and time $T'$ can be simulated by reversible Turing Machines of space $O(S'\cdot \log T')$ and time $\poly(T')$~\cite{bennett}. A different analysis of the algorithm of~\cite{bennett} shows that when $S'=\Theta(T')$ (which is indeed the case for our range of parameters) the algorithm can be reversibly simulated in space $O(S')$ and time $O(T')$~\cite{levinesherman}.~\footnote{This can also be seen directly by simply copying bits into fresh memory whenever they are erased. Since the time complexity of the original algorithm is comparable to the space complexity, this step is not costly.} In particular, it follows that $\ext^{(i)}$ can be computed by \emph{unitary} quantum algorithms in space $N\cdot \polylog (N)$ and time $N\cdot \polylog (N)$.  (The time complexity is with regards to computing any coordinate of the output.)

\paragraph*{The INW generator} The INW generator~\cite{inw} for space $S$ branching programs with $T$ input bits is defined recursively as follows. For $i\in \mathbb{N}$, the $i$-th generator $G_i:\{0,1\}^{N} \times \{0,1\}^{iN} \to\{0,1\}^{2^i}$ is defined at $x\in \{0,1\}^{N}$ and $s_1,\ldots,s_i \in\{0,1\}^N$ by
\[G_0(x)=x_1 \]
\[ G_i (x , s_1,\ldots,s_i) : = G_{i-1}(x,s_1,\ldots,s_{i-1}) \circ G_{i-1}(\ext^{(i)}_{s_i}(x, s_1,\ldots,s_{i-1})) .\]
Here, $x\circ y$ denotes the concatenation of strings $x$ and $y$. The generator $G_M$ naturally defines a binary tree of depth $M$ as follows. Consider a binary tree of depth $M$ where we number the layers from bottom to top, that is, the root has height $M$ and the leaves have height 1. We label each node by a string in $\{0,1\}^{N\cdot (M+1)}$ as follows. The root is labelled with the input $(x,s_1,\ldots,s_M)$ to $G_M$. Given a label $(x',s'_1,\ldots,s'_M)$ at any node at height $i\le M$ where $x',s'_1,\ldots,s'_M\in\{0,1\}^N$, the label at the left child is the same as that of its parent, while the label at the right child is $(\ext_{s'_i}^{(i)}(x',s'_1,\ldots,s'_{i-1}),s'_i,\ldots,s'_M)$.  The output of a leaf is defined to be the first coordinate of the label of the leaf. Let $j\in [T]$. Note that the binary expansion of $j$ defines a path from the root to the leaf of the tree. Observe that the leaf obtained by traversing the $j$-th path outputs the $j$-th coordinate of $G_M(x,s_1,\ldots,s_M)$.

Note that $G_M$ stretches $N\cdot (M+1)$ uniform bits to at least $T$ bits. This generator may produce more bits than necessary, but we may truncate output to the first $T$ bits. The inputs to $G_M$ are of length $O((S+ \log(1/\eps))\cdot \log T)$.  We refer to $G_M$ as the INW generator for space $S$ time $T$ algorithms. 

\begin{claim}\label{spacetimeclaim}
	The INW generator $G_M$ can be computed by  { unitary} quantum algorithms in space $O(M\cdot N)$ and time $\poly(M,N).$ That is, there is a  {unitary} quantum algorithm which given input $(x,s_1,\ldots,s_M)$ for $x,s_1,\ldots,s_M\in \{0,1\}^N$ and a coordinate $i\in [T]$, runs in space $O(M\cdot N)$ and time $M^2\cdot N^2\cdot \polylog (N)$ and outputs the $i$-th coordinate of $G(x,s_1,\ldots,s_M)$.
\end{claim}
\begin{proof} Let $(x,s_1,\ldots,s_M)$ be an input to $G_M$ where $x,s_1,\ldots,s_M\in \{0,1\}^N$. Consider the binary tree associated to the computation of $G_M$, as described before. Given any $j\in [T]$, we show how to simulate the process of traversing the $j$-th path in the tree using a unitary quantum algorithm. At each time step $i=1,\ldots,M$, we will ensure that the working memory essentially only contains the label of the $i$-th vertex in the $j$-th path. We now show how to update the memory to preserve this property.
	
Suppose we are at a node of height $i\le M$ and the current memory is $(x',s'_1,\ldots,s'_T)$ for some $x',s_1',\ldots,s'_T\in \{0,1\}^N$. Note that if $j_{M-i+1}=0$, then we don't have to update the memory. If $j_{M-i+1}=1$, we wish to update the memory to $(\ext^{(i)}_{s'_i}(x',s'_1,\ldots,s'_{i-1}), s'_i,s'_{i+1},\ldots,s_T')$. Firstly, note that this update operation is a reversible operation, in particular, it is its own inverse. This relies on our particular choice of extractors based on \textsc{xor} with $\delta$-biased spaces as in \Cref{fsextractor}. Secondly, recall that $\ext^{(i)}:\{0,1\}^{iN}\times \{0,1\}^N\to\{0,1\}^{iN}$ is computable by \emph{unitary} quantum algorithms in space and time at most $N\cdot \polylog (N)$. This implies that there is a unitary quantum algorithm which uses additional $N\cdot \polylog (N)$ space and $N\cdot \polylog (N)\cdot (iN)$ time and can update the memory from $(x',s'_1,\ldots,s'_T)$ to $(\ext_{s'_i}(x',s'_1,\ldots,s'_{i-1}), s'_i,s'_{i+1},\ldots,s'_T)$. Note that in particular, it returns any additional memory to the all zeroes state. Thus, the algorithm only requires $O(N\cdot M)+N\cdot \polylog (N)=O(N\cdot M)$ memory. The time complexity is $O(N\cdot \polylog (N)\cdot (iN))$ per iteration and $i$ varies from $1$ to $M$. This completes the proof.
\end{proof}

For a more refined bound on the time complexity of our simulation, we require the following claim.

\begin{claim}
	\label{refinedclaim} 
	Consider the binary tree associated with the INW generator $G_M$. Then, for every node $v$ and a neighbor $u$ of $v$, there is a unitary quantum algorithm which maps each possible label $\ell\in\{0,1\}^{N\cdot (M+1)}$ at $v$ to the induced label at $u$. Furthermore, if $v$ is at height $i$, then this algorithm runs in time $O(N\cdot \polylog (N)\cdot (iN))$ and space $O(N\cdot M)$.
\end{claim}
\begin{proof}[Proof of \Cref{refinedclaim}]
The proof of this follows from the proof of \Cref{spacetimeclaim}. Let $\ell=(\ell_0,\ell_1,\ldots,\ell_M)$ where  $\ell_0,\ldots,\ell_M\in\{0,1\}^N$.  Consider the special case when $v$ is the root.  In this case, the left child's label is simply $\ell$, while the right child's label is $\ext^{(M)}(\ell)$. This proves the claim for the root node. Suppose $v$ is an intermediate node, then \Cref{spacetimeclaim} demonstrates the desired algorithm for the children of $v$. To obtain the label of the parent of $v$, if $v$ was the left child of its parent, we return $\ell$, otherwise we return $(\ext^{(i-1)}_{\ell_{i-1}}(\ell_1,\ldots,\ell_{i-2}),\ell_{i-1},\ldots,\ell_M)$. The space complexity of this algorithm is $O( N\cdot M)$ to store the label plus $O(N\cdot \polylog (N))$ workspace. Overall, the space complexity is $O(N\cdot M)$. The time complexity is $O( N\cdot \polylog (N) )$ per output bit for the first $iN$ output bits (the rest of the bits are identical to those of the input). Thus, the overall time complexity is $O(N\cdot \polylog (N)\cdot (iN))$.
\end{proof}

\section{Main Result}

\label{sec:main}
We prove that the INW Generator fools quantum space-bounded branching programs which read classical random bits. 

\begin{theorem} \label{maintheorem1} Let $B^{\alpha}:\{0,1\}^T\rightarrow\Scal(\Hcal_S)$ be any space $S$ quantum branching program reading $T$ random bits with initial state $\alpha\in\Scal(\Hcal_S)$. Fix parameters $N=\Theta\pbra{ S + \log (1/\eps)}$ and $ M=\lceil\log T\rceil$ as before. Let $G_M:\{0,1\}^{(M+1)N}\rightarrow \{0,1\}^T$ be the INW generator as defined earlier. Then,
\[ \vabs{B^{\alpha}(U_T)- B^{\alpha} (G_M(U_{(M+1)N}))  }_1 \le \eps. \]
\end{theorem}

\cref{tracedistancedefinition} implies that the states $B^{\alpha}(U_T)$ and $B^{\alpha}( G_M(U_{(M+1)N}))$ cannot be distinguished by a measurement with more than $\eps/2$ advantage.  Thus, for $\eps=1/2$, we have that $G_{M}$ takes inputs of length $O(S\cdot  \log T)$ and outputs a random string of length $T$ that is indistinguishable from the uniform distribution over $T$ bits with more than $\frac{1}{4}$ probability by any quantum branching program of space $S$ and time $T$. We derive the following consequence of this. 

\begin{theorem} \label{maintheorem2} 
		Every quantum algorithm of time $T$ and space $S\ge \log T$ {\it with unitary operators and  intermediate measurements} can be simulated with error $\frac{1}{4}$ by a quantum algorithm with space ${O}(S\cdot\log T)$ and time $T\cdot S^2 \cdot \polylog (S)$ with unitary operators and {\it without intermediate measurements}. 
\end{theorem}

\section{Proof Of \cref{maintheorem1} }

\label{sec:maintheorem1}
The proof of \cref{maintheorem1} is immediate from the following lemma. 

\begin{lemma} \label{mainlemma} Let $i\in \mathbb{Z}_{\ge 0}$ and $x\in \{0,1\}^n$. Let $B^{\alpha}:=B_x^{\alpha}(\cdot):\{0,1\}^{2^i}\rightarrow \Scal(\Hcal_S)$ be any space $S$ quantum branching program with initial state $\alpha\in\Scal(\Hcal_S)$ and input $x\in \{0,1\}^n$ which reads $2^i$ bits of randomness.  Let $V,S_1,\ldots,S_i\sim U_{N}$ be independent samples. Then, for all $\alpha\in\Scal(\Hcal_S)$, 
 \[  \vabs{ B^{\alpha}(U_{2^i})- B^{\alpha}( G_i(V,S_1,\ldots,S_i )) }_1 \le \frac{3^i\cdot  \eps}{   T^2}.  \]
\end{lemma}

Setting $i=M\triangleq \lceil \log T\rceil$ in the above lemma implies that $\vabs{ B^{\alpha}(U_T)- B^{\alpha}( G_M(U_{(M+1)N})) }_1 \le  \eps$ and completes the proof of \cref{maintheorem1}. It suffices to prove the above lemma.

\begin{proof}[ Proof of \cref{mainlemma}]
The proof is by induction on $i$. The base case is true since $G_0(U_{N})=U_1$. Let us assume that the statement holds for all $j<i$  for all bounded space branching programs. Let $U_{2^{i}}=(U,U')$ where $U,U'\sim U_{2^{i-1}}$ are two independent uniformly distributed random variables on $2^{i-1}$ bits. We apply Triangle Inequality on the distance corresponding to the $i$-th generator as follows. 

Let $W=(V,S_1,\ldots,S_{i-1})\sim U_{iN}$ and $W'\sim U_{iN}$, independently of $W$. 
\begin{align}\begin{split}
&\vabs{ B^{\alpha}(U_{2^{i}}) - B^{\alpha}(G_i(V,S_1,\ldots,S_i)) }_1 \\
&\triangleq \vabs{ B^{\alpha}(U,U') - B^{\alpha}( G_{i-1}(W)\circ G_{i-1}(\ext_{S_i}(W )) }_1\\
&\le  \vabs{ B^{\alpha}(U,U')  - B^{\alpha}( G_{i-1}(W)\circ U') }_1 \\
& + \vabs{B^{\alpha}( G_{i-1}(W)\circ U')  - B^{\alpha}( G_{i-1}(W)\circ G_{i-1}(W')) }_1  \\
& +  \vabs{B^{\alpha}( G_{i-1}(W)\circ G_{i-1}(W'))  - B^{\alpha}( G_{i-1}(W)\circ G_{i-1}(\ext_{S_i}(W)) }_1 
\label{triangle} \end{split}\end{align}

We show that the each of these terms are bounded by $ \tfrac{3^{i-1}\cdot \eps }{T^2}$. Let the branching program $B^{\alpha}:\{0,1\}^{2^{i}}\rightarrow \Scal(\Hcal_S)$ be decomposed as $(B_0\cdot B_1)^{\alpha}$ where $B_0^{\alpha},B_1^{\beta}:\{0,1\}^{2^{i-1}}\rightarrow \Scal(\Hcal_S)$. Here, $B_0^{\alpha}(r_0)$ is a branching program that takes a (random) input $r_0\in \{0,1\}^{2^{i-1}}$ and runs the first half of $B^{\alpha}$ on the state $\alpha$, and $B_1^{B_0^{\alpha}(r_0)}(r_1)$ is a branching program that takes a (random) input $r_1\in \{0,1\}^{2^{i-1}}$ and runs the second half of $B^{\alpha}$ on the final state of the first half. 

We will first bound the first term by induction. By induction, the final state of the first branching program $B_0$ is nearly identical whether the inputs are drawn according to $ U_{2^{i-1}}$ or according to $G_{i-1}(W)$. More precisely, by induction, we have that
\[  \vabs{B_0^{\alpha}(U)- B_0^{\alpha}( G_{i-1}(W)) }_1  \le  \frac{3^{i-1}\cdot\eps}{ T^2 }.  \]
For any distribution $D\sim\{0,1\}^{2^{i-1}}$, the state $B_1^{B_0^{\alpha}(D)}(U')\triangleq B^{\alpha}(D,U')$ is obtained by taking the state $B_0^{\alpha}(D)$ and acting on it by the quantum operation $B_1(U')=\mathbb{E}_{r_1\sim U'}[B_1(r_1)]$. \cref{tracedistanceproperty} implies that quantum processing can only decrease the trace distance between two states. This, along with the above inequality gives us the following bound on the first term.
 \begin{equation}\label{firstbound}  
 	\vabs{ B^{\alpha}(U,U') -B^{\alpha}( G_{i-1}(W), U') }_1 \le  \frac{3^{i-1}\cdot \eps}{   T^2} . \end{equation}

The second term can be similarly bounded. Let $\beta=B_0^{\alpha}(G_{i-1}(W))$ be the final state of the first branching program $B_0^{\alpha}$ on the initial state $\alpha$ and input drawn according to the distribution $G_{i-1}(W)$. By the induction hypothesis applied on the second branching program $B_1^{\beta}$, the final state is nearly identical whether the inputs are drawn according to $U'$ or according to $G_{i-1}(W')$. That is,
\[ \vabs{ B_1^{\beta}(U')- B_1^{\beta}( G_{i-1}(W')) }_1 \le \frac{3^{i-1}\cdot \eps}{   T^2}.  \]
Since $B_1^{\beta}(D)=B^{\alpha}(G_{i-1}(W),D)$ for all distributions $D\sim\{0,1\}^{2^{i-1}}$, this gives us the following bound on the second term.
\begin{align}
\vabs{ B^{\alpha}(G_{i-1}(W),U')- B^{\alpha}(G_{i-1}(W), G_{i-1}(W')) }_1  \le\frac{3^{i-1}\cdot \eps}{   T^2}  \label{secondbound}
\end{align}

To bound the third term, we apply \cref{fsextractorcorollary} with the following parameters. Let $W$ be the uniform distribution over $\{0,1\}^{iN}$ as before. For all $w\in \{0,1\}^{iN}$, let $\rho_w:=B_0^{\alpha}(G_{i-1}(w))\in\Scal(\Hcal_S)$ be the final state reached by the first half of the program on the random string $G_{i-1}(w)$ and let $\rho_{XS}=\underset{w\sim W}{\E}\sbra{ \ket{w}\bra{w} \rho_w}$ where the system $X$ consists of $iN$ registers and the system $S$ consists of $S$ registers. Firstly, $\rho_{XS}$ is a classical-quantum state that is classical and uniform on $X$. Note that $t:=H_2(\rho_{XS}|S)\ge H_{\min}(\rho_{XS}|S) \ge H_{\min}(\rho_{XS})- S\ge iN-S$ due to \cref{conditioning} and \cref{mintwoentropy}. We chose $N=\Theta(\log T + S +\log(1/\eps) )\ge S + 2\log (iN) + 2\log(T^2/\eps)+O(1)\ge iN-t + 2\log (iN) + 2\log(T^2/\eps)+O(1)$. \cref{fsextractorcorollary} implies we have $\|\sigma-\sigma'\|_1\le \frac{\eps}{T^2}$, where
\[ \sigma:=   \underset{\substack{w\sim U_{iN}\\S_i\sim U_N} }{\E}\sbra{ \ket{\ext^{(i)}_{S_i}(w)}\bra{\ext^{(i)}_{S_i}(w)} B_0^\alpha (G_{i-1}(w)) }\quad\text{and}\quad \sigma':=  
\frac{\mathbb{I}}{2^{iN}}\otimes  \underset{w\sim U_{iN}}{\E} \sbra{B_0^\alpha (G_{i-1}(w)) } \]
Consider the controlled operator $\tilde{B}_1$ on $iN+S$ qubits defined by $\tilde{B}_1(\ket{w}\bra{w} \rho) = \ket{w}\bra{w}B_1^{\rho}(G_{i-1}(w))$ for  $w\in\{0,1\}^{iN}$ and $\rho\in\Scal(\Hcal_S)$. This operator has the effect of applying the quantum operation $B_1(G_{i-1}(w))$ on the last $S$ qubits, provided the first $iN$ qubits are in the state $\ket{w}\bra{w}$. Note that 
\[ \tilde{B}_1(\sigma)=  \underset{\substack{w\sim U_{iN}\\ S_i \sim U_N}}{\E}  \sbra{\ket{\ext^{(i)}_{S_i}(w)}\bra{\ext^{(i)}_{S_i}(w)} B^\alpha(G_{i-1}(w),G_{i-1}(\ext^{(i)}_{S_i}(w)))}\] 
\[ \tilde{B}_1(\sigma')=\underset{\substack{w\sim U_{iN}\\w'\sim U_{iN}}}{\E}\sbra{ \ket{w'}\bra{w'} B^\alpha(G_{i-1}(w),G_{i-1}(w'))}\]
Since $\| \sigma- \sigma'\|_1 \le  \frac{\eps}{T^2} $, \cref{tracedistanceproperty} implies that $\vabs{ \tilde{B}_1(\sigma)- \tilde{B}_1(\sigma')}_1\le \frac{\eps}{T^2} $. Note that for the state $\tilde{B}_1(\sigma)$, ignoring the first $iN$ qubits, the state on the last $S$ qubits is given by $B^{\alpha}(G_{i-1}(W),G_{i-1}(\ext_{S_i}(W)))$ while for $\tilde{B}_1(\sigma')$, it is given by $B^{\alpha}(G_{i-1}(W),G_{i-1}(W'))$. This implies that \begin{equation} \label{thirdbound} \vabs{ B^{\alpha}(G_{i-1}(W),G_{i-1}(\ext_{S_i}(W)))-B^{\alpha}(G_{i-1}(W),G_{i-1}(W'))}_1\le \frac{\eps}{T^2} .\end{equation}
Substituting \cref{firstbound}, \cref{secondbound} and \cref{thirdbound} in \cref{triangle} completes the inductive step. 
\end{proof}

\section{Proof of \cref{maintheorem2} }

\label{sec:maintheorem2}

We now demonstrate a proof of \Cref{maintheorem2} with the same space bound of $O(S\cdot \log T)$ but a weaker running time bound of $T\cdot (S\cdot \log T)^2\cdot \polylog (S)$. We will then show how to modify the algorithm so that the running time is $T\cdot S^2\cdot \polylog (S)$. 

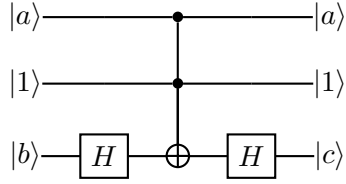
\begin{figure}
	\begin{quantikz}
		\ket{a} & \qw & \ctrl{2} &\qw & \qw \ket{a}\\
		\ket{1} & \qw & \ctrl{1} & \qw & \qw  \ket{1} \\
		\ket{b} & \gate{H} & \targ{} &\gate{H} &\qw \ket{c}\\
	\end{quantikz}
	\label{figure1}
	\centering
	\caption{Implementation of controlled-$U_1$ using Hadamard and Toffoli gates. Here,  $\ket{c}=\ket{U_1(b)} \text{ if }a=1 \text{ and } \ket{b}\text{ otherwise}$ for bits $a,b\in\{0,1\}$.}
\end{figure}

The proof of this theorem follows by an equivalence between quantum algorithms which perform intermediate measurements and quantum branching programs with classical randomness. Consider a qubit initialized to $\ket{0}$ and repeatedly apply the Hadamard operator and the measurement operator. The sequence of outcomes of the measurement operator is a uniformly random string. Thus, intermediate measurements allow quantum algorithms to simulate random coins. Conversely, intermediate measurements can be simulated using random coins as follows.

\begin{lemma}\label{randomcoinsmeasurement}
 Let $M$ refer to the measurement operator on the first qubit in the $\{\ket{0},\ket{1}\}$ basis. Let $U_0=\mathbb{I}$ be the identity operator and let $U_1$ be the reflection operator of the first qubit about $\ket{1}$. Then, for all $\rho\in \mathcal{S}(\mathcal{H})$, we have
		$M(\rho) = \frac{1}{2}\left( U_0\rho U_0 + U_1\rho U_1^\dagger\right)$.
\end{lemma}

\begin{proof}[Proof of \cref{randomcoinsmeasurement}]
	Let $\rho=\ket{i}\bra{j}\in \mathcal{S}(\mathcal{H})$ be a basis element of $\Scal(\Hcal)$ for $i,j\in \{0,1\}^*$. Note that 
	\begin{align*}
		\begin{split}
			\frac{1}{2}\left( U_0\rho U_0 + U_1\rho U_1^\dagger\right) &= \frac{1}{2}\left( \ket{i}\bra{j}+ (\mathbbm{1}_{i_1=j_1} - \mathbbm{1}_{i_1\neq j_1})\cdot \ket{i}\bra{j} \right)\\
			&= \mathbbm{1}_{i_1=j_1}\cdot  \ket{i}\bra{j} =M(\rho)
		\end{split}
	\end{align*}
Since the above equality holds for all basis elements $\rho\in \mathcal{S}(\mathcal{H})$ and both sides are linear in $\rho$, this equation holds for all $\rho\in \mathcal{S}(\mathcal{H})$. 
\end{proof}

\begin{proof}[Proof of \cref{maintheorem2}]

 \begin{figure}
 	\centering
\begin{quantikz}
	\lstick{$S$ space}  \ket{0} &  \qwbundle{}
	\qw   &\qw &\qw & \gate{U_1} &\qw &\qw &\gate{U_2}&\qw  \ldots &\meter{} &\qw \\
	\lstick{$1$ qubit}  \ket{0} & \qw & \qw &  \gate[wires=3][1cm]{\mathcal G_M^{(1)}} & \ctrl{-1} & \gate[wires=3][1cm]{{\mathcal G_M^{(1)}}^{-1}} & \gate[wires=3][1cm]{{\mathcal G_M^{(2)}}}  & \ctrl{-1}& \qw \ldots \\
	\lstick{${O}((S+\log T)\cdot \log T)$ \\ work space} \ket{0} &  \qwbundle{} \qw  &  \qw  & \qw & \qw & \qw &\qw  & \qw &\qw  \ldots  \\
	\lstick{$O((S+\log T)\cdot\log T)$ \\seed} \ket{0} &  \qwbundle{} \qw & \gate{H^{\otimes r}} &  &\qw   &\qw & & \qw & \qw \ldots & \\
\end{quantikz}
\label{figure2}
\caption{Simulation in the proof of \cref{maintheorem2}. Suppose the branching program $\Bcal$ consists of a sequence of controlled unitaries $U_1,\ldots,U_T$, controlled on a sequence of uniformly random bits. The figure describes the simulation of $\Bcal$ by a unitary algorithm.}
 \end{figure}
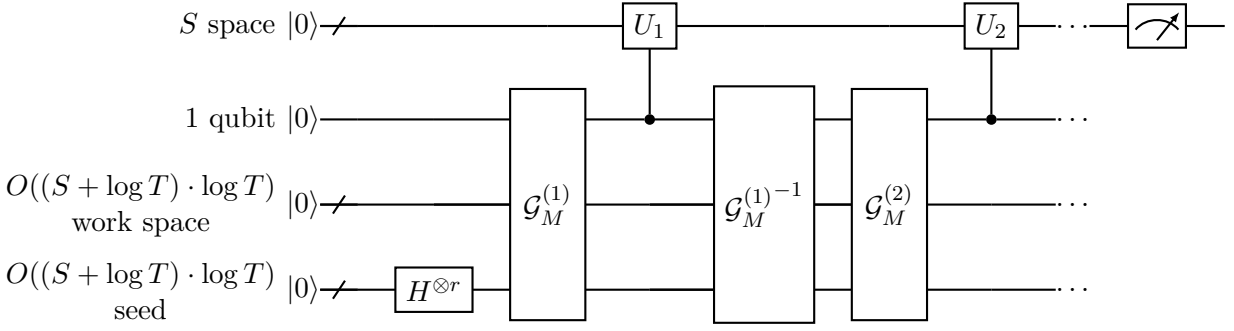

We begin by describing the simulation which runs in time $T\cdot (S\cdot \log T)^2\cdot \polylog ( S)$. Let $\mathcal{Q}$ be a quantum algorithm of time $T$ and space $S\ge \log T$ whose operators are in $\Gcal_U\cup\{R,M\}$. From \cref{randomcoinsmeasurement}, it follows that every measurement operator can be simulated by tossing a random coin, followed by applying controlled-$U_1$. The controlled-$U_1$ operator can be constructed using the Hadamard and the Toffoli gates as shown in Figure 1. Thus, the correspondence from \cref{randomcoinsmeasurement} defines a quantum branching program $\mathcal{B}$ such that: The initial state is the all zeroes state, the program uses (at most) $T$ bits of randomness and space $S+1$, the operators are in $\Gcal_U\cup\{R\}$, and most importantly, for all inputs $x\in \{0,1\}^n$, the expected output of $\Bcal_x$ when run on a uniformly random string in $\{0,1\}^T$ is identical to the final state of the algorithm $\mathcal{Q}$ on input $x$. Note that the operators of the branching program $\mathcal{B}$ are all unitary, furthermore, the sequence of operators is indeed logspace-uniform.

Let $M=\Theta(\log T)$ and $N=\Theta(S)$ be the parameters in the definition of the generator $G_M$ with $\eps=1/4$. \cref{spacetimeclaim} implies that $G_M$ is computable by unitary quantum algorithms in space ${O}(S\cdot \log T)$ and time $S^2\cdot (\log T)^2\cdot \polylog (S)$. That is, there is a unitary quantum algorithm (call it $\mathcal{G}_M^{(i)}$) which runs in space $O(S\cdot \log T)$ and time $S^2\cdot (\log T)^2\cdot \polylog (S)$ and which on input $r\in \{0,1\}^{N\cdot(M+1)}$ and a coordinate $i\in[T]$, outputs the $i$-th coordinate of $G_M(r)$. We use ${\mathcal G_M^{(i)}}^{-1}$ to refer to the inverse of the algorithm. 

Consider a quantum algorithm $\mathcal{Q}'$ of space $ {O}(S\cdot \log T)$ where the first $S$ qubits correspond to the memory of the branching program $\mathcal{Q}$, the next qubit is the random bit, the next ${O}(S\cdot \log T)$ qubits correspond to the workspace of the pseudorandom generator plus the space to store and iterate over coordinates in $[T]$ and the last $O(S\cdot \log T)$ qubits are the seed to the generator. The algorithm applies the Hadamard operator to the seed, which creates the uniform superposition of all possible seeds. It then computes the first output bit of $G_M$ on this superposition of seeds, simulates one step of the branching program using this bit, and uncomputes this bit. It similarly simulates all $T$ steps of the branching program. Finally, the algorithm measures the first qubit and outputs the outcome. Note that this algorithm exactly simulates the branching program $\mathcal{B}$ on input $x$ and a random string drawn from the output $G_M(U_{N(M+1)})$ of the generator.  
We now apply \cref{maintheorem1} to conclude that the output of the branching program $\mathcal{B}$ on inputs drawn from $U_T$ is indistinguishable with more than $\frac{1}{4}$ advantage from the case when the inputs are drawn from $G_{M}(U_{N(M+1)})$. Thus, the algorithm $\mathcal{Q}'$ described above simulates the original algorithm with error $\frac{1}{4}$. The overall space complexity is $O(S\cdot \log T)$ and the time complexity is $T\cdot S^2\cdot (\log T)^2\cdot \polylog (S)$. 

We now describe the modification required to achieve a time dependence of $O(T\cdot S^2\cdot \polylog  (S))$. We will use \cref{refinedclaim} as opposed to \cref{spacetimeclaim}. The main observation is that we don't need to uncompute at every time step, instead, we only need to uncompute to the point that we can compute the next coordinate. More precisely, consider the binary tree of depth $\lceil \log T\rceil$ associated with the computation of $G_M$ and consider performing a left-to-right DFS (depth-first-search) on this tree. If we consider the sequence of leaves reached by this DFS, the first coordinates of the labels give rise to the output $G_M(x,s_1,\ldots,s_M)$ (furthermore, the bits are in the correct order). Thus, performing a left-to-right DFS on this tree and computing the labels of the vertices along the DFS would allow us to sequentially compute the coordinates of the output $G_M(x,s_1,\ldots,s_M)$. To perform a DFS and compute the labels, we use the unitary maps given by \cref{refinedclaim}. It suffices to analyze the space and time complexity of this procedure. 

\cref{refinedclaim} implies that the space complexity of any individual walk step is at most $O(N\cdot M)$ and since we walk on a binary tree of depth $M$, the overall space complexity is at most $O(N\cdot M  +M)=O(N\cdot M)$. The time complexity at a vertex of height $i\in[M]$ is at most $N\cdot \polylog (N)\cdot (iN)$. Since there are at most $4T/2^i$ vertices at height $i$, it follows that the overall time complexity is at most 
\[  N\cdot \polylog (N)\cdot \sum_{i=1}^M 4T\cdot (iN)\cdot \frac{1}{2^i} \le  T\cdot N^2\cdot \polylog (N)\cdot .\]
This implies that the overall time complexity of the simulation is $T\cdot S^2\cdot\polylog (S)$. This completes the proof.
\end{proof}

\section{Acknowledgements}

We thank anonymous reviewers for  helpful comments and for pointing out the simulation of intermediate measurements using reset operations.

\appendix
\section{Appendix}
\subsection{Finite Field Arithmetic}

We use $\Fbb_2[x]$ to denote the set of univariate polynomials over $\Fbb_2$. This is a ring with the usual notions of addition and multiplication of polynomials. Furthermore, the addition of two polynomials of degree at most $m$ can be done in space $O(m)$ and time $O(m)$, while multiplication can be done in space and time at most $O(m\log m \log\log m)$. (See Theorem 8.23 from~\cite{vonzurgathen}.) Overall, arithmetic in this ring can be done in space and time at most $m\cdot\polylog (m)$. Let $f(x)\in \Fbb_2[x]$ be a polynomial of degree $m$. We use $\langle f(X)\rangle \triangleq \{g(x)\cdot f(x)| g(x)\in \Fbb_2[x]\}$ to denote the ideal generated by $f(x)$. The quotient space $\Fbb_2[x]/\left<f(x)\right>$ consists of cosets of the form $[g(x)]=g(x)+\langle f(x)\rangle$ where $g(x)$ is an arbitrary polynomial in $\Fbb_2[x]$ of degree at most $m-1$. This quotient space is a ring and inherits the addition and multiplication operations from $\Fbb_2[x]$, that is $[g(x)+h(x)]=[g(x)]+[h(x)]$ and $[g(x)\cdot h(x)]=[g(x)]\cdot [h(x)]$. The time and space complexity of division of polynomials of degree at most $m$ is bounded by that of multiplication. In particular, arithmetic in the ring $\Fbb_2[x]/\langle f(x)\rangle$ can be performed in space and time at most $m\cdot \polylog (m).$ (See Theorem 9.6 from~\cite{vonzurgathen}.) The polynomial $f(x)$ is irreducible if and only if the quotient space $\Fbb_2[x]/\left< f(x)\right>$ is a finite field. (See Theorem 1.61 in \cite{galoistextbook}.) Furthermore, if $f$ has degree $m$, then the corresponding field is of size $2^{m}$. It is well known that the polynomials $f_i(x):=x^{2\cdot 3^i}+x^{3^i}+1$ for $i\in \Nbb$ are irreducible over $\Fbb_2[x]$. (See Exercise 3.96 in \cite{galoistextbook}.)  Thus, given any $m\in \Nbb$ such that $m=2\cdot 3^i$, we can efficiently construct a finite field of characteristic two of size $2^m$ by considering the quotient space $\Fbb_2[x]/\langle f_i(x) \rangle $. Futhermore, addition and multiplication of two elements in this field can be done in space and time $m\cdot \polylog (m)$. Similarly, raising elements of the field to $k$-th powers can be done in space $m\cdot \polylog (m)+\log k$. and time $m\cdot \polylog (m)\cdot \log k$ by repeated squaring.


\begin{thebibliography}{1}
	
	
	
	
	\bibitem[AGH+90]{alonetal} Simple Constructions of Almost k-wise Independent Random Variables - Noga Alon, Oded Goldreich, Johan Hastad, Rene Peralta. FOCS 1990:  544-553
	
	
	
	\bibitem[Ben89]{bennett}
	Time/Space Trade-offs for Reversible Computation - Charles H. Bennett. SIAM J. Comput. 18(4): 766-776 (1989)
	
	
		 
	\bibitem[BNS89]{bns}
	L\'{a}szl\'{o} Babai, Noam Nisan, Mario Szegedy:
	Multiparty Protocols, Pseudorandom Generators for Logspace, and Time-Space Trade-Offs.  STOC 1989: 1-11
	
	
	\bibitem[BRRY10]{brry}
	Mark Braverman, Anup Rao, Ran Raz, Amir Yehudayoff:
	Pseudorandom Generators for Regular Branching Programs. FOCS 2010: 40-47
	
	
	
	
	\bibitem[FR21]{fr}
	Eliminating Intermediate Measurements in Space-Bounded Quantum Computation - Bill Fefferman, Zachary Remscrim. STOC 2021
	
	\bibitem[FS08]{fs}
	Randomness Extraction Via $\delta$-Biased Masking in the Presence of a Quantum Attacker - Serge Fehr, Christian Schaffner. TCC 2008: 465-481
	
	
	
	\bibitem[GG99]{vonzurgathen}
	Modern Computer Algebra, 3rd edition - 
	Joachim von zur Gathen, Jurgen Gerhard. CanadaISBN 9781139856065
	
	\bibitem[GR14]{gr}
	Anat Ganor, Ran Raz:
	Space Pseudorandom Generators by Communication Complexity Lower Bounds. APPROX-RANDOM 2014: 692-703
	
	
	\bibitem[GRZ21]{grz}
	Quantum Logspace Algorithm for Powering Matrices of Bounded Norm - Uma Girish, Ran Raz, Wei Zhan. ICALP 2021
	
	\bibitem[INW94]{inw}
	Pseudorandomness for Network Algorithms - Russel Impagliazzo, Noam Nisan, Avi Wigderson. STOC 1994: 356-364
	
	
	
	\bibitem[{Lan}61]{Landauer}
	R.~{Landauer}:
	Irreversibility and heat generation in the computing process.
	{  IBM Journal of Research and Development}, 5(3):183--191, 1961.
	
	
	
	\bibitem[LN83]{galoistextbook}
	Introduction to Finite Fields and their Applications - 
	Rudolf Lidl, Harald Niederreiter, ISBN 9781139172769
	
	\bibitem[LS90]{levinesherman}
	A Note on Bennett's Time-Space Tradeoff for Reversible Computation - Robert Y. Levin, A. Sherman
	
	
	\bibitem[MRT19]{mrt}
	Raghu Meka, Omer Reingold, Avishay Tal:
	Pseudorandom generators for width-3 branching programs. STOC 2019: 626-637
	
	
	
	\bibitem[{Nis}90]{nis}
	Noam Nisan:
	Pseudorandom generators for space-bounded computation.  STOC 1990: 204-212
	
	
	\bibitem[NN90]{naornaor}
	Small-bias probability spaces: efficient constructions and applications - Joseph Naor, Moni Naor. STOC 1990: 213-223 
	
	\bibitem[NZ96]{nz}
	Noam Nisan, David Zuckerman:
	Randomness is Linear in Space. J. Comput. Syst. Sci. 52(1): 43-52 (1996)
	
	
	\bibitem[{Pip}79]{nick}
	On Simultaneous Resource Bounds - Nicholas Pippenger. FOCS 1979: 307-311
		
	\bibitem[{Ren}08]{rennerthesis}
	Security of Quantum Key Distribution - Renato Renner. International Journal of Quantum Information 6(01):  1-127 (2008), Ph.D. Thesis
	
	
	\bibitem[RR99]{rr}
	Ran Raz, Omer Reingold:
	On Recycling the Randomness of States in Space Bounded Computation. STOC 1999: 159-168
	
	\bibitem[RSV13]{rsv}
	Omer Reingold, Thomas Steinke, Salil P. Vadhan:
	Pseudorandomness for Regular Branching Programs via Fourier Analysis. APPROX-RANDOM 2013: 655-670
	
	
	
	
	
	
	
	
	
	
	\bibitem[{Shi}03]{shi}
	Both Toffoli and controlled-Not need little help to do universal quantum computation - Y. Shi. Quantum Information \& Computation 3(1): 84–92 (2003)
	
	
	

	
 	\bibitem[{Wat}99]{watrous}
 	Space-Bounded Quantum Complexity - John Watrous.  J. Comput. Syst. Sci. 59(2): 281-326 (1999)
 	
 	\bibitem[{Yao}93]{yao} Quantum Circuit Complexity - Andrew Yao. FOCS 1993: 352-361
 	
\end{thebibliography}
\end{document}